\documentclass[11pt]{article}
\usepackage{times}
\usepackage{amsmath,amsfonts}
\usepackage{epsfig,amssymb,amstext,xspace,amsthm}
\usepackage{graphicx,color}
\usepackage{verbatim}
\usepackage{multirow}
\usepackage[shortlabels]{enumitem}
\usepackage{comment}
\usepackage{calc,bbm}

\newtheorem{theorem}{Theorem}[section] 
\newtheorem{lemma}[theorem]{Lemma}

\newtheorem{proposition}[theorem]{Proposition}
\newtheorem{observation}[theorem]{Observation}

{\theoremstyle{remark} 

{\theoremstyle{definition}
\newtheorem{algorithm}{Algorithm}}
 
\newtheorem{definition}[theorem]{Definition}}

\newcommand{\bg}[1]{\medskip\noindent{\it #1}}
 
\newenvironment{proofbox}[1][Proof]{\begin{proof}[#1]}{\end{proof}}

\newenvironment{proofof}[1]{\begin{proof}[Proof of {#1}]}{\end{proof}}

\newcommand{\supp}{{\rm supp}}

\newcommand{\tsp}{\ensuremath{\mathsf{TSP}}\xspace}
\newcommand{\atsp}{\ensuremath{\mathsf{ATSP}}\xspace}
\newcommand{\atspp}{\ensuremath{\mathsf{ATSPP}}\xspace}
\newcommand{\dl}{\ensuremath{\mathsf{DirLat}}\xspace}

\newcommand{\reg}{\ensuremath{\mathsf{reg}}}

\begin{document}

\title{A Constant-Factor Approximation for Directed Latency in Quasi-Polynomial Time}
\author{ 
    Zachary Friggstad\thanks{{\tt zacharyf@ualberta.ca}.  
    Dept. of Computer Science, Univ. Alberta, Edmonton, AB T6G 2E8.
    Supported by the Canada Research Chairs program and an NSERC Discovery grant.}  
\and
    Chaitanya Swamy\thanks{{\tt cswamy@uwaterloo.ca}.  
    Dept. of Combinatorics and Optimization, Univ. Waterloo, Waterloo, ON N2L 3G1. 
    Supported in part by NSERC grant 327620-09 and an NSERC Discovery Accelerator
    Supplement Award.}
}
\date{} 

\date{}

\maketitle


\begin{abstract}
We give the first constant-factor approximation for the Directed Latency
problem in quasi-polynomial time. Here, we must visit all clients in an asymmetric metric using a single vehicle starting at a depot $r$.
This should be done in a way that minimizes the average time a node waits to be visited by the vehicle. The approximation guarantee is an improvement over the polynomial-time $O(\log n)$-approximation [Friggstad, Salavatipour, Svitkina, 2013]
and no better quasi-polynomial time approximation algorithm was known.

To obtain this, we must extend a recent result
showing the integrality gap of the Asymmetric TSP-Path LP relaxation is bounded by a constant [K\"{o}hne, Traub, and Vygen, 2019], which itself builds on the breakthrough result that the integrality gap for standard Asymmetric TSP is also a constant [Svensson, Tarnawsi, and Vegh, 2018]. We show the standard Asymmetric TSP-Path integrality gap is bounded by a constant even if the cut requirements of the LP relaxation are relaxed from $x(\delta^{in}(S)) \geq 1$ to $x(\delta^{in}(S)) \geq \rho$ for some constant  $1/2 < \rho \leq 1$. We also give a better approximation guarantee in the special case of Directed Latency in regret metrics where the goal is to find a path $P$ minimize the average time a node $v$ waits in excess of $c_{r,v}$, i.e. $\frac{1}{|V|} \cdot \sum_{v \in V} (c_v(P)-c_{r,v})$.
\end{abstract}


\section{Introduction}

We investigate the {\em Directed Latency} problem (\dl), a vehicle routing problem where we are to route a single vehicle to serve a set of clients/nodes.
Unlike the standard Traveling Salesman problem (\tsp) where the goal is to minimize the length of the route, in \dl the goal is to minimize the average
time a node waits to be served. Formally, in \dl we are given an asymemtric metric space $(V \cup \{r\}, c)$ where $V$ is a set of node nodes, $r$ is the depot node, and $c$ gives asymmetric metric distances
over $V \cup \{r\}$.
That is, $c_{u,v} \geq 0$ for any two nodes $u,v$, $c_{u,u} = 0$ for any node $u$, and $c_{u,v} \leq c_{u,w} + c_{w,v}$ for any three nodes $u,v,w$. Our goal is to find a Hamiltonian path $P$ starting at the
depot $r$ to minimize $\sum_{v \in V} c_P(v)$ where $c_P(v)$ denotes the total cost of all edges on the $r-v$ subpath of $P$. This sometimes called the {\em Traveling Repairman} problem.

Our main contribution is the first constant-factor approximation for \dl in quasi-polynomial (i.e. $n^{O(\log n)})$ time.
A key technical contribution towards this is generalizing recent work by K\"{o}hne, Traub, and Vygen \cite{KohneTV19} to give constant-factor integrality gap bounds for a slight weakening of the standard LP relaxation
for Asymmetric TSP-Path (\atspp).
We also get explicit constants for the special case of \dl in so-called {\em regret metrics} where the goal is to minimize $\sum_{v \in V} (c_P(v) - c_{r,v})$ when $(V,c)$ is a symmetric metric (i.e. $c_{u,v} = c_{v,u}$).
That is, we want to minimize the average time each node waits in excess of their shortest-path distance from $r$.
This can be cast as special case of \dl by using {\em regret distances} $c^{\text{reg}}_{u,v} := c_{r,u} + c_{u,v} - c_{r,v}$, which form an asymmetric metric.

The algorithm we present is based on a time-indexed linear programming (LP) relaxation, much like the approach taken in \cite{PostS15} for the {\em Minimum Latency} problem in symmetric metrics. Roughly speaking, our approach uses variables for $(v,t)$ pairs where $v \in V$ is a node to be visited and $t$ is the time they should be visited. Other variables indicate transitions between nodes at different times.



\subsection{Related Work}
Nagarajan and Ravi first studied \dl and obtained an approximation guarantee of $n^{1/2 + \epsilon}$ in time $n^{O(1/\epsilon)}$ for any constant $\epsilon > 0$ \cite{NagarajanR08}, which extends easily to an $O(\alpha' \cdot \log^{O(1)}(n))$-approximation in quasi-polynomial time where (roughly speaking) $\alpha'$ is an upper bound on the integrality gap of the natural Held-Karp LP relaxation for \atspp.
They also showed $\alpha'$ is bounded by $O(\sqrt n)$.
Friggstad, Salavatipour, and Svitkina improved the approximation guarantee for \dl and the upper bound on the integrality gap for \atspp to $O(\log n)$ \cite{FriggstadSS13}.
This is currently the best polynomial-time approximation for \dl and no better quasi-polynomial time approximation was known before our work. If the metric is symmetric, constant-factor approximations are know. The first
was given by Blum et al. \cite{Blum94}, the best guarantee so far is a 3.59-approximation by Chaudhuri et al. \cite{Chaudhuri03}.

Post and Swamy studied LP relaxations for the undirected minimum latency problem \cite{PostS15}. Using time-indexed LP relaxations, they obtain improved approximations for the multi-depot variant and also recover the 3.59-approximation 
for the single-vehicle version using an LP relaxation. We build off ideas behind one of their LP relaxations in this work.

The integrality gap for \atspp has seen some improvements since \cite{FriggstadSS13}. In \cite{FriggstadGS16}, it is shown the integrality gap is in fact $O(\log n / \log \log n)$. Recently, \cite{KohneTV19} shows the integrality gap is in fact $O(1)$.
Specifically, they show the gap is at most $4 \cdot \alpha - 3$ where $\alpha$ is the integrality gap for the Held-Karp relaxation for standard ATSP. Prior to this, Svensson, Tarnawski, and Vegh showed $\alpha$ is bounded by a constant \cite{SvenssonTV18}. An even more recent development by Traub and Vygen improves the bound to $\alpha \leq 22$ \cite{TV2020}.
Currently, the best lower bound on $\alpha$ is 2 \cite{CharikarGK06}.



\subsection{Results and Techniques}
Our main result is the following. Throughout, we let $n$ denote $|V|$.
\begin{theorem}\label{thm:dlqpt}
For some constant $c \geq 1$, there is a $c$-approximation for \dl running in time $n^{O(\log n)}$ time.
\end{theorem}

To discuss this, we first introduce some notation.
For a directed graph $G = (V,E)$ and some $S \subseteq V$, we let $\delta^{in}_G(S) = \{(u,v) \in E : u \in V-S, v \in S\}$, $\delta^{out}_G(S) = \{(u,v) \in E : u \in S, v \in V-S\}$ and $\delta_G(S) = \delta^{in}_G(S) \cup \delta^{out}_G(S)$. If the graph is clear from the context, we may omit the subscript $G$. We often identify an asymmetric metric $(V \cup \{r\}, c)$ with the complete directed graph over nodes $V \cup \{r\}$ having edge costs $c_{u,v}$ for distinct $u,v \in V \cup \{r\}$.
For a path $P$ and a node $v$ on $P$, let $c_P(v)$ be the cost of the $r-v$ subpath of $P$.

We first scale the distances in the metric be polynomially-bounded integers. Standard scaling techniques allow us to do this.
\begin{theorem}\label{thm:scaling}
For any constant $\epsilon > 0$,  if there is an $\alpha(n)$-approximation for instances of \dl where each $c_{u,v}$ is a positive integer bounded by a polynomial in $n$ and $1/\epsilon$ and where $c_{u,v} \geq 1$ for nodes $u \neq v$, then there
is an $(\alpha(n)+\epsilon)$-approximation for general instances of \dl.
\end{theorem}

So we may assume all distances $c_{u,v}$ are integers bounded as such.
Let $T = n \cdot \max_{u,v} c_{u,v}$ and notice that $T$ is bounded by a polynomial in $n$.
Any Hamiltonian path in the metric $(V \cup \{r\}, c)$ has length at most $T$, so all nodes in the optimum solution are visited by time $T$. For brevity, let $[T] = \{0, 1, \ldots, T\}$.


We begin with essentially the same time-indexed LP relaxation that was used in \cite{PostS15} for the undirected minimum latency problem, specifically (LP3) in their work.
The variables in the time-indexed relaxation are the following: for $v \in V\cup\{r\}$ and $t \in [T]$ let $x_{v,t}$ indicate that we
visit $v$ at time {\em exactly} $t$, let $z_{e,t}$ indicate we finished traversing edge $e$ at time exactly $t$. This is slightly different than \cite{PostS15} where they let $z_{e,t}$ indicate $t$ was traversed {\em by} time $t$.
Note, we omit Constraints (14) from \cite{PostS15}, one can easily show they are implied by our slightly different approach.
\begin{alignat}{3}
{\bf minimize}: \quad & \sum_{v \in V, t \in [T]} t \cdot x_{v,t} \tag{LP-Latency} \label{lp:dl} \\
{\bf subject~to}:
\quad & \sum_{t \in [T]} x_{v,t} \quad & = &\quad 1 &&\quad \forall~v \in V \label{cons:enough} \\
\quad & \sum_{e \in \delta^{in}(S)} \sum_{t' \leq t} z_{e,t'} \quad & \geq & \quad \sum_{t' \leq t} x_{v,t'}  &&\quad \forall~v \in V, ~S \subseteq (V \cup \{r\}), t \in [T] \label{cons:cut_t} \\
\quad & x_{v,t} = \sum_{e \in \delta^{in}(v)} z_{e,t}  \quad & \geq &\quad \sum_{e \in \delta^{out}(v)}z_{e,t} &&\quad \forall~{v \in V, t \in [T]} \label{cons:flow} \\
& x, z \quad &\geq &\quad 0 \nonumber
\end{alignat}
It is easy to check that an optimal solution $P^*$ naturally corresponds to an integral solution to \eqref{lp:dl} with the same cost as the latency of $P^*$.
The constraints admit an efficient separation oracle simply by checking for each $v \in V$ and $t \in T$ if the minimum $r-v$ cut has capacity at least $\sum_{t' \leq t} x_{v,t'}$ when using a capacity of $\sum_{t' \leq t} z_{e,t'}$ for each edge $e$.

Our proof of Theorem \ref{thm:dlqpt} proceeds by bucketing clients based on their fractional latencies, finding low-cost paths for these buckets, and stitching these paths together to form our final path. Our advantage over \cite{FriggstadSS13} comes from the fact that we guess the $O(\log T) = O(\log n)$ nodes $v^*_i$ appearing at distances roughly $2^i$ along the optimum path $P^*$, plus
their exact visiting $\ell^*_i$ times along $P^*$. We add constraints to \eqref{lp:dl} to reflect these guesses. For each $v^*_i$, consider the nodes $v$ that are at least, say, $2/3$-visited before $v^*_i$ is visited: call this the {\em bucket} $B_i$ for $v^*_i$. With a bit of modification, the restriction of
\eqref{lp:dl} to the times before $\ell^*_i$ is visited induces an LP solution with cost $O(2^i)$ for the natural \atspp LP relaxation that covers all $v \in B_i$ to an extent of at least 2/3. That is, we get a solution to the following LP relaxation for \atspp for $\rho = 2/3$.
\begin{alignat}{3}
{\bf minimize}: \quad & \sum_{uv} c_{u,v} \cdot x_{u,v} & & \tag{$\textsc{LP-ATSPP}_\rho$} \label{lp:atspp} \\
{\bf subject~to}: \quad & x(\delta^{out})(v) - x(\delta^{in})(v) \quad &= &\quad \left\{\begin{array}{rl}-1 & v = s \\ +1 & v = t \\ 0 & v \neq s,t \end{array}\right.&&\quad \forall~ v \in V \nonumber \\
& x(\delta(U)) \quad &\geq &\quad 2 \cdot \rho && \quad \forall~ \emptyset \subsetneq U \subseteq V-\{s,t\} \nonumber \\
& x \quad &\geq &\quad 0 \nonumber
\end{alignat}
The integrality gap of the case $\rho = 1$ was proven to be constant in \cite{KohneTV19}. At this point, we need a stronger integrality gap bound.
\begin{theorem}\label{thm:atspp_gap}
For some absolute constant $c$ that is independent of $\rho$, the integrality gap of \eqref{lp:atspp} is at most $\frac{c}{2\rho-1}$.
\end{theorem}

In \cite{FriggstadSS13}, it was shown that if $\rho = 1/2$ then the integrality gap of \eqref{lp:atspp} is unbounded even if we strengthen it to have an in-flow of 1 for each $v \in V-\{s,t\}$ (but still have the relaxed cut constraints). 
As a side note, we also show the dependence on $\rho$ is asymptotically correct as $\rho$ approaches 1/2.
\begin{theorem}\label{thm:lower}
There is an instance of \atspp where the integrality gap of \eqref{lp:atspp} on that instance is $\geq \frac{1}{2\rho-1}$ for every $1/2 < \rho \leq 1$ even if we strengthen the LP with constraints $x(\delta^{in}(v)) = 1$ for each $v \in V-\{s,t\}$.
\end{theorem}

Returning to the idea behind the proof of Theorem \ref{thm:dlqpt}, once we have these paths $P_i$ we must bound the cost of stitching the last node of $P_i$ to the first node after $r$ on $P_{i+1}$. This is where guessing plays the most prominent role, we show that strengthening the LP with our guess ultimately implies this new edge used to stitch $P_i$ to $P_{i+1}$ has cost $O(2^i)$, as required.

Our final result is an improved approximation in the case that the metric is the regret metric of an undirected metric, which we simply call {\em regret metrics}.
\begin{theorem}\label{thm:regret}
The integrality gap of \eqref{lp:atspp} in regret metrics is at most $\alpha^{\text{reg}}_\rho := \frac{300}{42-12\sqrt6} \cdot \frac{1}{2\rho-1} \approx \frac{23.8}{2\rho-1}$ and we can find a path $P$ whose cost is at most $\alpha^{\text{reg}}_\rho$ times the value of an optimum LP solution.
\end{theorem}
We then work out an explicit constant for approximating \dl in regret metrics.
\begin{theorem}\label{thm:regret_dl}
There is a quasi-polynomial time 778-approximation for \dl in regret metrics.
\end{theorem}
While this constant is large, it it considerably better than what we would obtain if we simply used Theorem \ref{thm:atspp_gap} and the current-best bound on $\alpha$, which would lead to an approximation guarantee in the tens of thousands.

~

\noindent
{\bf Outline of the Paper}\\
Section \ref{sec:dl} proves Theorem \ref{thm:dlqpt} and discusses how Theorem \ref{thm:regret_dl} would follow from Theorem \ref{thm:regret}.
The scaling result itself (Theorem \ref{thm:scaling}) is fairly standard, it's proof is found in Appendix \ref{app:scaling}.
Section \ref{sec:atspp} proves Theorem \ref{thm:atspp_gap}. Theorem \ref{thm:regret_dl} is proven in Section \ref{sec:regret}.
The example from Theorem \ref{thm:lower} appears in Appendix \ref{sec:lower}.


\section{An $O(1)$-Approximation in Quasi-Polynomial Time}\label{sec:dl}

Recall, by Theorem \ref{thm:scaling}, we may assume distances are integers bounded by a polynomial in $n$ and that $c_{u,v} \geq 1$ for distinct nodes $u,v$.
We also let $T = n \cdot \max_{u,v \in V \cup \{r\}} c_{u,v}$, which is an upper bound on the cost of any Hamiltonian path. We focus on a fixed optimal path $P^*$. Our algorithm starts by guessing the last node $v^*_i$ visited by $P^*$ at some time in the interval\footnote{One can show the geometric factor of 2 is optimal for our analysis for any $1/2 < \rho < 1$, so we fix it now.}
$[2^i, 2^{i+1})$ (if any)
and its exact distance $\ell^*_i \in [T]$ for each $0 \leq i \leq \log_2 T = O(\log n)$. Let $v^*_i = \bot$ if no such node exists for this interval. For any $i$, we then know that no node is visited at any time in $[2^i, 2^{i+1})$ if $v^*_i = \bot$ and, if $v^*_i \neq \bot$, we also know no node is visited at a time in the interval $(\ell^*_i, 2^{i+1})$ so we mark these times as {\bf forbidden}. Let $A = \{i : v^*_i \neq \bot\}$ be {\bf admissible} buckets corresponding to intervals where the optimum visits at least one node.
Let $1/2 < \rho \leq 1$ be a parameter we optimize later.

{\small  \vspace{5pt}

\begin{algorithm}[{\bf Directed Latency: $O(1)$-approximation in $n^{O(\log n)}$ time}]  \label{alg:dlqpt}
\vspace{-5pt} \hrule 

~

\noindent
Input: asymmetric metric $(V\cup\{r\},c)$ with integer distances at most $T/n$.

\noindent
Output: an $r$-rooted path $P$
\end{algorithm}
\smallskip
\begin{enumerate}[label=D\arabic*., ref=D\arabic*, topsep=0ex, itemsep=0.5ex,
    leftmargin=*]
\item \label{dguess}
For every choice (guess) of $v^*_i \in V \cup \{\bot\}$ for each $0 \leq i \leq \log_2 T$ and $\ell^*_i \in [T]$ for each such $i$
where $v^*_i \neq \bot$, perform the following steps. Let $F = \{t  \in [T]: t \in [2^i, 2^{i+1}) \text{ where } v^*_i = \bot \text{ or } t \in (\ell^*_i, 2^{i+1}) \text{ where } v^*_i \neq \bot\}$ be the forbidden times for this guess $(v^*, \ell^*)$ and $A = \{i \in [0, \log_2 T] : v^*_i \neq \bot\}$ the admissible buckets.

\begin{enumerate}[label*=\arabic*., ref=\theenumi.\arabic*, topsep=0ex, itemsep=0ex,
    leftmargin=*]
\item \label{step:strengthen}
Get an optimal extreme point solution $(x,y,z)$ to \eqref{lp:dl} strengthened with the following additional constraints: 1) $x_{v^*_i, \ell^*_i} = 1$ for each $i \in A$ and 2) $x_{v,t} = 0$ for each $v \in V$ and $t \in F$.
If the LP is infeasible, abort this guess of $(v^*, \ell^*)$.
\item For each $v \in V$, let $t(v)$ be the minimum time such that $\sum_{t \leq t(v)} x_{v,t} \geq \rho$. For $i \in A$, let $B_i = \{v \in V : t(v) \in [2^i, 2^{i+1})\}$.
\item For each $i \in A$, use the algorithm from Theorem \ref{thm:atspp_gap} to get an $r-v^*_i$ path $P_i$ spanning $\{r\} \cup B_i$. \label{step:fracpath}
\item Let $P^{v^*, \ell^*}$ be the path obtained by concatenating the paths $\{P_i\}_{i \in A}$ in increasing order of $i$, and shortcutting past repeat occurrences of $r$.
\end{enumerate}

\item Return the best path $P^{v^*, \ell^*}$ found over all guesses where the strengthening of \eqref{lp:dl} was feasible.
\end{enumerate}
\hrule}

~

Let $P^*$ be an optimum solution and consider the iteration where $(v^*, \ell^*)$ is consistent with $P^*$.
Let $(x, z)$ be an optimum LP solution for the strengthening of \eqref{lp:dl} by the constraints in Step \eqref{step:strengthen}. Clearly this strengthened LP is feasible and the value of the solution $(x,z)$ is at most $OPT$, the latency of $P^*$.

For each $v \in V$, note that $t(v)$ is well-defined by Constraints \eqref{cons:enough}.
Ultimately, we will show the path $P^{v^*, \ell^*}$ visits each $v \in V$ by time $O(t(v))$.
We begin by showing this suffices to get a constant-factor approximation.
\begin{lemma}\label{lem:suffice}
Let $P$ be a path and $c \geq 1$ be such that $c_P(v) \leq c \cdot t(v)$ for each $v \in V$. Then the latency of $P$ is at most $\frac{c}{1-\rho} \cdot OPT$.
\end{lemma}
\begin{proof}
Fix some $v \in V$. By definition of $t(v)$, $\sum_{t(v) \leq t \leq T} x_{v, t} \geq 1-\rho$ which yields
$t(v) \leq \frac{1}{1-\rho} \cdot \sum_{t(v) \leq t \leq T} t(v) \cdot x_{v,t} \leq \frac{1}{1-\rho} \cdot \sum_{t \in [T]} t \cdot x_{v,t}$.
So, $\sum_{v \in V} c \cdot t(v) \leq \frac{c}{1-\rho} \sum_{v \in V} \sum_{t \in [T]} t \cdot x_{v,t} \leq \frac{c}{1-\rho} \cdot OPT$.
\end{proof}

\subsection{Bounding the Latency of $P^{v^*, \ell^*}$}

In the remainder of the proof it is convenient to view a ``time-expanded'' graph $G_T$. The nodes are pairs $(v,t)$ with $v \in V \cup \{r\}$
and $t \in [T]$ and an edge connects $(u,t)$ to $(v,t')$ if $c_{u,v} = t'-t$. Observe $G_T$ is acyclic. We can then view $z_{e,t}$ as assigning values to edges of $G_T$: the
edge $(u,t-c_{u,v}),(v,t)$ has value $z_{(u,v),t}$ and cost $c_{u,v}$.

The constraints of \eqref{lp:dl} mean $z$ constitutes one unit of $(r,0)$-preflow in $G_T$. Let $i'$ be the greatest index in $A$. Considering the LP constraints added in Step \eqref{step:strengthen}, we see $x_{v^*_{i'}, \ell^*_{i'}} = 1$
and $x_{v,t} = 0$ for all $t > \ell^*_{i'}$. Thus, $z$ must be a flow with value 1 in $G_T$ ending at $(v^*_{i'}, \ell^*_{i'})$.
Since the support of the flow $z$ is acyclic in $G_T$
and since one unit of flow passes through {\em every} $(v^*_i, \ell^*_i)$ node in $G_T$ for each $i \in A$, no flow skips past node $(v^*_i, \ell^*_i)$. That is, no edge $(u,t),(v,t')$ in $G_T$ supports any $z$-flow if $t < \ell^*_i < t'$ for some $i \in A$, nor does any edge $(u,t),(v,t')$ support any $z$-flow
if $t = \ell^*_i$ yet $u \neq v^*_i$ or $t' = \ell^*_i$ yet $v \neq v^*_i$ for some $i \in A$.


We start by showing we can compute low-cost paths covering each bucket.
First, we recall a famous splitting-off theorem by Mader. The following is a slight specialization of one such result.
\begin{theorem}[Mader \cite{Mader82}]\label{thm:mader}
Let $D = (V \cup \{s\}, A)$ be an Eulerian, directed graph with, perhaps, parallel edges such that the $u-v$ connectivity for every $u,v \in V$ is at least $k$. Then for every $(u,s) \in A$ there is some $(s,v) \in A$ such that in the graph $D' = (V \cup \{s\}, A-\{(u,s), (s,v)\} \cup \{(u,v)\})$, the $u-v$ connectivity for every $u,v \in V$ remains at least $k$.
\end{theorem}

For brevity, let $\alpha_\rho$ denote the integrality gap of \eqref{lp:atspp}.
\begin{lemma}\label{lem:path}
For each $i \in A$, we can compute a Hamiltonian $r-v^*_i$ path $P_i$ in $G[\{r\} \cup B_i]$ with cost $\alpha_\rho \cdot 2^{i+1}$ in polynomial time.
\end{lemma}


\begin{proof}
It suffices to show the optimal solution to \eqref{lp:atspp} in $G[\{r\} \cup B_i]$ (starting at $r$ and ending at $v^*_i$) has value at most $2^{i+1}$.
If so, then by Theorem \eqref{thm:atspp_gap} we can then efficiently find a Hamiltonian $r-v^*_i$ path $P_i$ in $G[\{r\}\cup B_i]$
with cost at most $\alpha_\rho\cdot 2^{i+1}$.

To that end, let $x'$ be a vector over edges of the metric given by $x'_{u,v} = \sum_{t < 2^{i+1}} z_{(u,v),t}$ for $u,v \in V \cup \{r\}$.
As discussed above, the truncation of $z$ to times $\leq 2^{i+1}$ constitutes one unit of flow from $(r,0)$ to $(v^*_i, \ell^*_i)$ in $G_T$,
so $x'_{uv}$ is then one unit of $r-v^*_i$ flow in the metric. Further, since the cost of an edge $(u,t-c_{u,v}),(v,t)$ is $c_{u,v}$ in $G_T$, the cost of this flow $x'$ is, in fact, exactly $\ell^*_i$ which is at most $2^{i+1}$.

Next we verify $x'(\delta(S)) \geq 2 \cdot \rho$ for each $S \subseteq V-\{v^*_i\}$ with $S \cap B_i \neq \emptyset$.
Consider some $v \in S \cap B_i$. Constraint \eqref{cons:cut_t}, the fact that $v \in B_i$, and the fact that $x_{v,t} = 0$ for $\ell^*_i < t < 2^{i+1}$ shows $x'(\delta^{in}(S)) = \sum_{e \in \delta(S)} \sum_{t < 2^{i+1}} z_{e,t} \geq \rho$.
Since $x'$ is an $r-v^*_i$ flow and $r,v^*_{i} \notin S$, then flow conservation shows $x'(\delta(S)) \geq 2 \cdot \rho$.

Much like in \cite{AnKS15} for the \textsc{Prize-Collecting TSP-Path} problem, one can use Theorem \ref{thm:mader} to shortcut $x'$ past nodes not in $B_i \cup \{r\}$ to get solution for \eqref{lp:atspp}
for in the graph $G[\{r\} \cup B_i]$ (with start node $s = r$ and end node $t = v^*_i$), also with cost at most $2^{i+1}$. 
That is, we may assume $x'$ is rational as $z$ is a rational vector (being part of an extreme point of an LP with rational coefficients).
Let $\Delta$ be an integer such that the vector $\Delta \cdot x'$ is integral. Consider the graph $G'$ with nodes $V \cup \{r\} \cup \{r'\}$ where $r'$ is a new node. The edges of $G'$ consist of $\Delta \cdot x'_{uv}$ copies of edge $uv$ for each $u,v \in V \cup \{r\}$, and $\Delta$ edges from $v^*_i$ to $r'$ and also from $r'$ to $r$ (each having cost 0). Note the $r-u$ connectivity for each $u \in V$ is at least $\Delta \cdot \rho$.
Note, the cost of all edges in $G'$ is at most $\Delta \cdot 2^{i+1}$.

For each $v \in V-B_i$, we iteratively perform the splitting off procedure from Theorem \ref{thm:mader} for $s = v$. The total cost of the edges does not increase by the triangle inequality (note the edges that are removed and added all lie in the metric over $V \cup \{r\}$), and the $r-u$ connectivity remains at least $\Delta \cdot \rho$ for each $u \in B_i$. After doing this for each $v \in V-B_i$, we are left with a multigraph of total edge cost cost no more than the total cost of all edges in $G'$. Further, if we remove all $v^*_ir'$ and $r'r$ edges, we still get the connectivity from $r$ to any other $v \in B_i$ is at least $\Delta \cdot \rho$. If $k_{uv}$ denotes the number of copies of $uv$ in this new graph, setting $x''_{uv} = k_{uv}/\Delta$ for each $uv \in G[\{r\} \cup B_i]$ yields a feasible LP solution for \eqref{lp:atspp} in the metric graph over $B_i \cup \{r\}$ (with start node $r$ and end node $v^*_i$) with cost at most $2^{i+1}$. Note that we do not actually need to perform this step in our algorithm, this analysis is simply proving the existence of a low-cost solution to 

By Theorem \eqref{thm:atspp_gap}, we can then efficiently find a Hamiltonian $r-v^*_i$ path $P_i$ in $G[\{r\}\cup B_i]$
with cost at most $\alpha_\rho\cdot 2^{i+1}$.
\ref{lp:atspp}.
\end{proof}


Next we bound the cost of stitching together the paths for the admissible buckets.
\begin{lemma}\label{lem:stitch}
Let $P_i$ and $P_{i'}$ be two paths constructed in Step \eqref{step:fracpath} for consecutive indices $i,i' \in A$. Let $u_{i'}$
be the first node on $P_{i'}$ after $r$ and recall $v^*_i$ is the last node of $P_i$. Then $c_{v^*_i,u_{i'}} \leq 2^{i'+1}$.
\end{lemma}
\begin{proof}
Note that $u_{i'} \in B_{i'}$ means $t(u_{i'}) \in [2^{i'}, 2^{i'+1})$. Also, $x_{u_{i'}, t(u_{i'})} > 0$ by definition of $t(u_{i'})$.
All units of $z$-flow in the acyclic graph $G_T$ pass through $(v^*_i, \ell^*_i)$ and also through $(v^*_{i'}, \ell^*_{i'})$. So the restriction of $z$ to edges $(u,t),(v,t')$ in $G_T$ with $\ell^*_i \leq t \leq t' \leq \ell^*_{i'}$ constitutes one unit of $(v^*_i, \ell^*_i)-(v^*_{i'}, \ell^*_{i'})$ flow that supports $(u_{i'}, t(u_{i'}))$. Therefore, a path decomposition of this restriction of $z$ includes $(u_{i'}, t(u_{i'}))$ on some path.
Any such path has cost exactly $\ell^*_{i'} - \ell^*_i \leq 2^{i'+1}$. By the triangle inequality, $c_{v^*_i, u_{i'}} + c_{u_{i'}, v^*_{i'}} \leq 2^{i'+1}$.
\end{proof}

Next, we bound the latency of each $v \in V$ along the final $P^{v^*, t^*}$ obtained by concatenating the $P_i$ paths for increasing indices $i \in A$ and shortcutting past all but the first occurrence of $r$.
\begin{lemma}\label{lem:latency}
$d_{P^{v^*, \ell^*}} \leq 4(\alpha_\rho+1) \cdot t(v)$ for any $v \in V$.
\end{lemma}
\begin{proof}
Consider any  $v \in V$ and say it lies on $P_i$. To reach $v$ along $P^{v^*, \ell^*}$, we traverse paths $P_{i'}$ for $i' < i$ plus the ``stitching'' edges $v^*_{i'}u^*_{i''}$ for consecutive indices $i',i'' \in A$, $i'' \leq i$. By Lemma \eqref{lem:path} and Lemma \eqref{lem:stitch}, the latency of $v$ along $P^{v^*, \ell^*}$ can be bounded by
$\sum_{i' \in A, i' \leq i} \alpha_\rho \cdot 2^{i'+1} + \sum_{i' \in A, i' \leq i} 2^{i'+1} \leq (\alpha_\rho+1) \cdot\sum_{i'=0}^i \cdot 2^{i+1} \leq 4(\alpha_\rho+1)\cdot 2^i \leq 4(\alpha_\rho+1) \cdot t(v)$.
\end{proof}

Set $\rho = 2/3$ and note Theorem \ref{thm:atspp_gap} implies $\alpha_{2/3}$ is bounded by a constant.
The proof of Theorem \ref{thm:dlqpt} then follows readily from Lemmas \eqref{lem:suffice} and \eqref{lem:latency} and the fact that $T$ is bounded by a polynomial in $n$.

Using our approach even with the improved bound of $\alpha \leq 22$ from \cite{TV2020}  produces an approximation ratio in the tens of thousands using our framework.
We conclude by demonstrating a much better constant-factor guarantee in the special case of regret metrics.
\begin{proofof}{Theorem \ref{thm:regret}}
Choosing $\rho = 0.74743$ and using the integrality gap bound from Theorem \ref{thm:regret} yields $\alpha_\rho \leq 48.09442$ in this regret metrics.
Then using Lemmas \ref{lem:latency} and \ref{lem:suffice} and choosing $\epsilon$ sufficiently small in Theorem \ref{thm:scaling} yields a 778-approximation.
\end{proofof}


\section{Bounding the Integrality Gap of \eqref{lp:atspp}}\label{sec:atspp}

Consider nodes $V$ with two distinguised $s,t \in V$ and asymmetric metric distances $c_{u,v}$ between points of $V$. We consider \eqref{lp:atspp} for the Asymmetric TSP Path problem
where the goal is to find the cheapest Hamiltonian $s-t$ path. As mentioned earlier, the integrality gap is unbounded if $\rho \leq 1/2$ \cite{FriggstadSS13}, so we focus on the case $1/2 < \rho \leq 1$.
As in \cite{KohneTV19}, we start withthe dual of \eqref{lp:atspp}.
\begin{alignat}{3}
{\bf maximize}:\quad & z_t - z_s + \sum_U 2\rho \cdot y_U & & \tag{$\textsc{DUAL}_\rho$} \label{lp:dual} \\
{\bf subject~to}: \quad & z_v - z_u + \sum_{U : uv \in \delta(U)} y_U \quad & \leq &\quad c_{u,v} &&\quad  \forall~ u,v \nonumber \\
& y \quad & \geq &\quad 0 \nonumber
\end{alignat}
Naturally, our proof borrows many steps from K\"{o}hne, Traub, and Vygen \cite{KohneTV19} but there are additional challenges we have to work through in this more general setting.

For a vector $x$ over the edges $E$ of the directed metric (when viewed as a complete, directed graph), let $\supp(x) = \{uv \in E: x_{u,v} > 0\}$.
Similarly, for a vector $y$ over cuts of the metric let $\supp(y) = \{\emptyset \subsetneq S \subseteq V-\{s,t\} : y_S > 0\}$. From now on, we focus on the graph $G = (V,\supp(x))$.
The proofs of Propositions \ref{prop:uncross}, \ref{prop:cuts}, and \ref{prop:cross} are very similar to proofs in \cite{KohneTV19} and are omitted or just sketched in this paper.
\begin{proposition}\label{prop:uncross}
Given any optimal dual solution $(y,z)$, one can find an optimal dual solution $(y', z)$ with $\supp(y')$ being laminar
in polynomial time.
\end{proposition}
In other words, we can modify $y$ to be laminar without changing $z$ using efficient uncrossing techniques. The proof is exactly the same as the proof in \cite{KohneTV19} essentially because the set of feasible solutions
to \eqref{lp:dual} does not change if we select different $\rho$.

The next proposition is almost identical to one in \cite{KohneTV19}, but we omit the case $U = V$ in the statement. In fact, the result may not be true for this case $U = V$, we handle that separately below.
\begin{proposition} \label{prop:cuts}
Let $x$ be an optimum primal solution and let and $G = (V, \supp(x))$.
For any $U \subseteq V-\{s,t\}$ with $x(\delta(U)) = 2\rho$, any topological ordering $U_1, \ldots, U_\ell$ of the strongly connected components of $G[U]$ satisfies:
\begin{itemize}
\item $\delta^{in}(U_1) = \delta^{in}(U)$,
\item $\delta^{out}(U_\ell) = \delta^{out}(U)$, and
\item $x(\delta^{out}(U_i) = \delta^{in}(U_{i+1}))$ for any $1 \leq i < \ell$.
\end{itemize}
\end{proposition}
We sketch the proof of Proposition \ref{prop:cuts} so the reader is assured it holds, though the proof is essentially the same.
\begin{proof}[Proof sketch]
Because $U$ is a tight set, $x(\delta^{in}(U)) = \rho$. Further, $x(\delta^{in}(U_1)) \geq \rho$. All edges in $\supp(x)$ entering $\delta(U_1)$ must lie in $\delta^{in}(U)$ because $U_1$ is the first node in the topological ordering.
Thus, $\rho = x(\delta^{in}(U)) \geq x(\delta^{in}(U_1)) \geq \rho$, so equality must hold throughout and $\delta^{in}(U) = \delta^{in}(U_1)$ as we are working in the support of $x$. A similar statement shows $\delta^{out}(U_\ell) = \delta^{out}(U)$.

For $i > 1$ we note $\delta^{in}(U_i) \subseteq \delta^{in}(U) \cup \bigcup_{j < i} \delta^{out}(U_j)$ simply because the $U_j$ are topologically ordered.
Inductively, we have $x(\delta^{out}(U_{i-1})) = \rho$ and each edge in $\delta^{in}(U) \cup \bigcup_{j < i-1} \delta^{out}(U_j)$ is already
proven to lie in $\delta^{in}(U_{j'})$ for some $j' < i$. So we see $\delta^{in}(U_i) \subseteq \delta^{out}(U_{i-1})$ and, thus,
\[ \rho = x(\delta^{in}(U_{i-1})) = x(\delta^{out}(U_{i-1})) \geq x(\delta^{in}(U_i)) \geq \rho. \]
So, again, equality must hold throughout.
\end{proof}

We use a different observation to address the case $U = V$ that was omitted from Proposition \ref{prop:cuts}. Intuitively, we show that it is still possible to buy a cheap set of edges to chain the strongly-connected components of $G$ in sequence but the cost of these edges does increase relative to $OPT_{LP}$ as $\rho \rightarrow 1/2$.
\begin{proposition}\label{prop:new}
In any topological ordering $U_1, \ldots, U_\ell$ of the strongly connected components of $G$, for each $1 \leq i < \ell$ there is some edge $uv \in \delta^{out}(U_i) \cap \delta^{in}(U_{i+1})$ with $c_{u,v} \leq \frac{1}{2\rho-1} \cdot \sum_{uv \in \delta^{out}(U_i) \cap \delta^{in}(U_{i+1}))} c_{u,v} x_{u,v}$.
\end{proposition}
\begin{proof}
This is easy for $i = 1$ and $i = \ell-1$. For example, we have $x(\delta^{in}(U_2) \geq \rho$ and all edges from $\delta^{in}(U_2)$ lie in $\delta^{out}(U_1)$. Thus, $x(\delta^{out}(U_1) \cap \delta^{in}(U_2)) \geq \rho$ so the cheapest edge in $\delta^{out}(U_1) \cap \delta^{in}(U_2)$ has cost at most $\frac{1}{\rho} \cdot \sum_{uv \in \delta^{out}(U_i) \cap \delta^{in}(U_{i+1}))} c_{u,v} x_{u,v}$. We finish by observing $1/\rho \leq 1/(2\rho-1)$ as $\rho \leq 1$. A similar argument works for $i = \ell-1$, so we now assume $1 < i < \ell-1$.

We quickly introduce notation. For an index $1 \leq j \leq \ell$ let $U_{\leq j} = \cup_{1 \leq j' \leq j} U_{j'}$ and $U_{\geq j} = \cup_{j \leq j' \leq \ell} U_{j'}$.
Let $\delta(X;Y)$ denote $\{uv \in \supp(x) : u \in X, v \in Y\}$ for $X,Y \subseteq V$. With this notation, let $a = x(\delta(U_i;U_{i+1})), b = x(\delta(U_i; U_{\geq i+2})),
c = x(\delta(U_{\leq i-1}; U_{i+1}))$, and $d = x(\delta(U_{\leq i-1}; U_{\geq i+1}))$. We have $a+b+c+d = x(\delta^{out}(U_{\leq i})) = 1$ as $\delta^{out}(U_{\leq i})$ is the disjoint union of the sets defining $a,b,c,d$).
On the other hand, $\rho \leq x(\delta^{out}(U_i)) = a+b$ and $\rho \leq x(\delta^{in}(U_{i+1})) = a+c$. Therefore, $2\rho - 1 \leq (a+b) + (a+c) - (a+b+c+d) \leq a$ so $x(\delta^{out}(U_i) \cap x(\delta^{in}(U_i)) \geq 2\rho-1$.
So the cheapest edge $(u,v) \in \delta^{out}(U_i) \cap \delta^{in}(U_{i+1})$ has $c_{u,v} \leq \frac{1}{2\rho-1} \cdot \sum_{(u',v') \in \delta^{out}(U_i) \cap \delta^{in}(U_{i+1}))} c_{u',v'} x_{u',v'}.$
\end{proof}

\begin{proposition} \label{prop:cross}
Let $G$ be the support graph of an optimum solution $x$ to \eqref{lp:atspp} and $(y,z)$ an optimum dual with $\supp(y)$ laminar. For any $U \in \supp(y) \cup \{V\}$
and any $u,w \in U$ with $w$ being reachable from $u$ in $G[U]$, there is a $v-w$ path in $G[U]$ that crosses each set $U' \in \supp(y)$ at most twice for $U' \subsetneq U$.
\end{proposition}
Again, the proof is the same as that in \cite{KohneTV19} which only relies on Proposition \ref{prop:cuts} for $U \in \supp(y)$ (i.e. not on the case $U = V$ that we omitted from the proposition in our setting). We sketch the argument
briefly to ensure the reader this still holds with the omission of $U = V$ from Proposition \ref{prop:cuts}.
\begin{proof}
Consider any $u-w$ path $P$ contained in $G[U]$. Suppose $U' \in \supp(y)$ is maximal among all such sets where $P$ re-enters $U'$ after it exits $U'$. Let $a$ be the first node of $P$ in $U'$ and $b$ the last node of $P$ in $U'$ (it could be $a = u$ or $b = v$).
Inductively, replace the $a-b$ portion of $P$ with an $a-b$ path in $G[U']$ that enters and leaves every set $U'' \in \supp(y)$ at most once for $U'' \subsetneq U'$. Repeat for all such maximal $U' \in \supp(y)$.
\end{proof}

\subsection{Constructing the Path}
Let $OPT_{LP}$ denote the optimum solution value to \eqref{lp:atspp}. Recall we let $\alpha$ denote an upper bound on the integrality gap of the standard Held-Karp relaxation for ATSP.
We will prove the following lemma later.
\begin{lemma}\label{lem:zbound}
An optimal dual solution $(y, z)$ with $\supp(y)$ being laminar and $z_s - z_t \leq \frac{1}{2\rho-1} \cdot OPT_{LP}$ can be computed in polynomial time.
\end{lemma}

Using this, we now turn to the main result of this section. Note, we are choosing simplicity in presentation over optimizing the constants in the guarantee.

\begin{proof}[Proof of Theorem \ref{thm:atspp_gap}]
Complementary slackness ensures every $U \in \supp(y)$ satisfies $x(\delta(U)) = 2\rho$.
Consider the edge support graph $G = (V, \supp(x))$. Modify $G$ to get an ATSP instance $H$ by adding a new node $\overline{v}$ and edges $(t,\overline v)$ with cost $OPT_{LP}$
and $(\overline v, s)$ with cost $0$.

It is easy to check that setting
\[ x'_{u,v} = \left\{\begin{array}{rl}
\frac{1}{\rho} & \text{if } (u,v) \in \{(t,\overline v), (\overline v, s)\} \\
\frac{x_{u,v}}{\rho} & \text{otherwise }
\end{array}\right.
\]
yields a feasible solution for the \atsp-Circuit relaxation from \cite{SvenssonTV18} in instance $H$ with cost $\frac{2}{\rho} OPT_{LP}$. Using \cite{SvenssonTV18}, we can find a circuit $W$ spanning all nodes in $H$ with cost at most
$\frac{2\alpha}{\rho}OPT_{LP}$ in polynomial time. This circuit must use the $(t,\overline v)$ edge at least once as it visits $\overline v$.
By deleting occurrences of $(t, \overline v)$ and $(\overline v, s)$, we get $s-t$ walks $W_1, \ldots, W_k$ in $G$ that collectively span all nodes in $V$
with $\sum_j c(W_j) \leq \frac{2\alpha}{\rho} \cdot OPT_{LP} \leq 4\alpha \cdot OPT_{LP}$.
We also point out $k \leq 4\alpha$
because in removing the $k$ edges incident to $\overline v$ to get the walks $W_i$, we removed a total edge cost of $k \cdot OPT_{LP}$ from a circuit whose cost is at most $4\alpha \cdot OPT_{LP}$,
so $k \leq 4\alpha$.

Let $U_1, \ldots, U_\ell$ be the strongly connected components of the support graph $G$. 
For each $U_i$, let $\mathcal W_i = \{j : W_j \text{ visits a node in } U_i\}$ and note $|\mathcal W_i| \leq k$. Unlike the case $\rho = 1$ in \cite{KohneTV19}, it could be that $j \notin \mathcal W_i$ for some $U_i$ and $W_j$. For each $1 \leq i \leq \ell$ and each $j \in \mathcal W_i$, let $R_{i,j}$ denote the restriction of $W_j$ to $U_i$.
Now, if some $W_j$ enters $U_i$, then once it leaves it cannot re-enter because $U_i$ is a strongly connected component of $G$.
So $R_{i,j}$ is a single walk for each $j \in \mathcal W_i$. For such $(i,j)$, let $u^i_j$ and $v^i_j$ be the first and last nodes of $W_j$ in $U_i$.

Order $\mathcal W_i$ as $j_1 < j_2 < \ldots < j_{|\mathcal W_i|}$. 
By Proposition \ref{prop:cross} and the fact each $U_i$ is a strongly connected component, we can find paths $P_{i,j_m}$ for $j_m \in \mathcal W_i$ from $v^i_{j_m}$ to $u^i_{j_{m+1}}$ (or $u^i_1$ if $m = |\mathcal W_i|$)
where $P_{i,j}$ enters and exits each $U' \in \supp(y)$ with $U' \subsetneq U_i$ at most once and does not cross any other set in $\supp(y)$.
Then, for each $i$ we get a circuit $C_i$ spanning all nodes of $U_i$ by adding the paths $P_{i,j}$ for $j \in \mathcal W_i$ to the walks $R_{i,j}$.

By Proposition \ref{prop:new}, for each $1 \leq i < \ell$ there are edges $u'_iv'_{i+1} \in \delta^{out}(U_i) \cap \delta^{in}(U_{i+1})$ with cost at most $\frac{1}{2\rho-1}$ times the fractional cost of edges in $\delta^{out}(U_i) \cap \delta^{in}(U_{i+1})$.
Also, say $v'_1 = s$ and $u'_\ell = t$. By fully traversing each $C_i$ starting at $v'_i$ and then continuing to follow it again to reach $u'_i$, we get $v'_i-u'_i$ walks $W'_i$ spanning $U_i$.
The final path $P$ we output is the concatenation of the walks $W'_1, W'_2, \ldots, W'_\ell$.
Let $S = \{v'_iu'_{i+1} : 1 \leq i < \ell\}$ be the edges used to ``stitch'' these walks $W'_i$ together.

To bound the cost of $P$,
first observe $c(S) \leq \frac{1}{2\rho-1} OPT_{LP}$ as the sets $\delta^{out}(U_i) \cap \delta^{in}(U_{i+1})$ are disjoint for $1 \leq i < \ell$.
To bound the cost of the cycles $C_i$, we define a modified cost $c^y_{uv} = \sum_{U : uv \in \delta(U)}$ and observe $c(Q) = z_v - z_u + c^y(Q)$ for any $u-v$ path $Q$ (the $z$-values for internal nodes of $Q$ cancel).

By complementary slackness, $c_{u,v} = z_v-z_u + c^y_{uv}$ for each $uv \in \supp(x)$.
Each $C_i$ was formed by stitching together endpoints of $R_{i,j}$ using paths $P_{i,j}$. Each $P_{i,j}$ crosses each $U' \in \supp(y), U \subsetneq U_i$ at most twice and does not cross any set in $\supp(y)$ not contained in $U_i$.
Further, no two $P_{i,j}, P_{i',j'}$ paths for $i \neq i'$ can cross the same $U' \in \supp(y)$ because the two paths are contained in different components of $G$.

Therefore, each $U' \in \supp(y)$ is crossed by at most $k$ paths of the form $P_{i,j}$
meaning $\sum_{i,j} c^y(P_{i,j}) \leq \sum_{i,j} z_{v^i_j}-z_{u^i_j} + 2k \cdot \sum_U y_U$. We also have $c^y(R_{i,j}) =  z_{u^i_j}-z_{v^i_j} + c(R_{i,j})$. Therefore, $\sum_i c^y(C_i) = \sum_i \sum_{j \in \mathcal W_{i,j}} c^y(P_{i,j}) + c^y(R_{i,j}) \leq 2k \sum_U y_U + \sum_{i,j \in \mathcal W_i} c(R_{i,j}) \leq 2k \sum_U y_U + \sum_j c(W_j)$ (the $z$ terms for the enpoints of the $R_{i,j}$ cancel out in the first inequality).

But $c(C) = c^y(C)$ for any cycle $C$ because, again, the $z$-terms cancel out. So
\[
\begin{array}{rrcll}
& c(P) \leq c(S) + 2 \cdot \sum_i c(C_i) & \leq & \frac{OPT_{LP}}{2\rho-1} + 2 \sum_{i=1}^k c(W_i) + 2k \sum_{U} y_U \\
 \leq & \frac{OPT_{LP}}{2\rho-1} + 4\alpha \cdot OPT_{LP} + 2k \sum_U y_U  & \leq & O(1)\cdot \frac{1}{2\rho-1} \cdot OPT_{LP}+  \frac{k}{\rho}\left(OPT_{LP} + z_s - z_t\right) \\
\leq & O(1) \cdot \frac{1}{2\rho-1} \cdot OPT_{LP} + \frac{k}{\rho} \cdot (z_s - z_t).
\end{array}
\]
Here, $O(1)$ refers to some constant that is independent of $\rho$ and we also recall $k$ is bounded by an absolute constant as well. Using Lemma \ref{lem:zbound} to bound $z_s - z_t$ finishes the proof.
\end{proof}


\section{Bounding $z_s-z_t$}
We prove Lemma \ref{lem:zbound} to finish the proof of Theorem \ref{thm:atspp_gap}. Our approach is more direct than \cite{KohneTV19}, they used an argument that shifts
LP weight around to show that $y_U > 0$ implies $U$ is not an $s-t$ separator in the support graph $G = (V, \supp(x))$. We establish this fact using complementary slackness
applied to the LP used to find the optimal solution to \ref{lp:dual} with minimum possible $z_s - z_t$. We comment that their proof could also be adapted to show what we want, we are presenting
this alternative proof because we feel it is more naturally motivated: we already want to minimize $z_s-z_t$ among all optimal duals so it is natural to ask what complementary slackness gives for $y_U > 0$.
\begin{proof}[Proof of Lemma \ref{lem:zbound}]
Let $x$ be an optimal primal solution to \ref{lp:atspp}.
Note that if we restricted the variables of \eqref{lp:atspp} and the constraints of \eqref{lp:dual} to $\supp(x)$ then $x$ and $(y,z)$ remains optimal.
For any feasible solution $(y,z)$ to \eqref{lp:dual}, we know $z_t - z_s \leq OPT_{LP}$ because $y \geq 0$. So the following LP is bounded. Note, we first solved \eqref{lp:atspp} to compute $OPT_{LP}$
which is then a fixed value (not a variable) in \ref{lp:dual2} below.
\begin{alignat}{3}
{\bf maximize}: \quad & z_t - z_s & & \tag{$\textsc{DUAL}_\rho\textsc{-z}$} \label{lp:dual2} \\
{\bf subject~to}:\quad & z_t - z_s + \sum_{\emptyset \subsetneq U \subseteq V-\{s,t\}} 2\rho \cdot y_U \quad& \geq &\quad OPT_{LP} \label{cons:dualcost}\\
 \quad & z_v - z_u + \sum_{U : uv \in \delta(U)} y_U \quad& \leq &\quad c_{u,v} && \quad\forall~ u,v \in \supp(x) \label{cons:edgecons} \\
& y \quad& \geq &\quad 0 \nonumber
\end{alignat}
The second constraint asserts $(y,z)$ is a feasible solution for \eqref{lp:dual}, so the first constraint then asserts it is an optimal solution for \ref{lp:dual}
In fact, in any feasible solution the first constraint must hold with equality.
We prove $z_s - z_t \leq \frac{1}{2\rho-1} \cdot OPT_{LP}$ for an optimal solution $(y,z)$ to \eqref{lp:dual2}.
With this, we finish the proof of Lemma \ref{lem:zbound} by
simply noting that Proposition \ref{prop:uncross} shows we can uncross the support of $y$ while leaving $z$ unchanged.

The LP that is dual to \eqref{lp:dual2} has a variable $\kappa$ for Constraint \eqref{cons:dualcost} of \eqref{lp:dual2} and new variables $x'_{uv}$ for each instance $uv$ of Constraint \eqref{cons:edgecons}.
\begin{alignat}{3}
{\bf minimize}: \quad & \sum_{uv \in \supp(x)} c_{u,v} \cdot x'_{uv} - OPT_{LP} \cdot \kappa& & \notag \\
{\bf subject~to}: \quad & x'(\delta^{out}(v)) - x'(\delta^{in}(v)) \quad & = &\quad \left\{\begin{array}{rl} 1+\kappa & v=s \\ -1-\kappa & v=t \\ 0 & v \neq s,t\end{array}\right. &&\quad \forall~ v \in V \nonumber \\
& x'(\delta(U)) \quad&\geq &\quad 2\rho \cdot \kappa &&\quad \forall~ \emptyset \subsetneq U \subseteq V-\{s,t\} \nonumber \\
& x', \kappa \quad&\geq &\quad 0 \nonumber
\end{alignat}
\begin{lemma}\label{lem:contract}
In an optimal solution $(y,z)$ to \ref{lp:dual2}, if $y_U > 0$ then there is an $s-t$ path in the graph $G[V-U]$.
\end{lemma}
\begin{proof}
Let $x'$ be an optimal solution to the dual of \eqref{lp:dual2}. Then $y_U > 0$ implies $x'(\delta(U)) = 2\rho \cdot \kappa$ so, by flow conservation, $x'(\delta^{in}(U)) = \rho \cdot \kappa$.

On the other hand, $x'$ constitutes an $s-t$ flow of value $1+\kappa$. Consider a decomposition of $x'$ into paths and cycles.
The total weight of paths that do not enter $U$ is at least $1+\kappa - \rho \cdot \kappa = 1 + (1-\rho) \cdot \kappa > 0$. Thus, there is an $s-t$ path in $G$ that does not pass through $U$.
\end{proof}

Continuing as in \cite{KohneTV19}, let $U_1, \ldots, U_k$ be the maximal sets in $\supp(y)$. In the graph $G'$ obtained by contracting each $U_i$, we have by Lemma \ref{lem:contract} that
for each contracted node $U_i$ there is an $s-t$ path in $G'$ that avoids $U_i$.
By a variant of Menger's Theorem (Lemma 9 in \cite{KohneTV19}), there are node-disjoint $s-t$ paths $P_1, P_2$ in $G'$.
Consider the edges of $P_1$ and $P_2$ in $G$. For any $U_i$, at most one of $P_1$ or $P_2$ enters (and exits) $U_i$. Suppose it is the case that one of them $\overline P \in \{P_1, P_2\}$ enters $U_i$.
Let $u, v$ be the first and last nodes of $\overline P$ as it passes through $U_i$. By Proposition \ref{prop:cross}, we can find a $u-v$ path in $G[U_i]$ that crosses each $U' \in \supp(y)$ contained in $U$ at most twice, and does not cross any other set in $\supp(y)$. Add these edges to $\overline P$.

Do this for each $U_i$ that is entered by some $\overline P \in \{P_1, P_2\}$. We get paths $P'_1, P'_2$ using only edges in $\supp(x)$ that, collectively, cross each set in $\supp(y)$ at most twice. Thus,
$0 \leq c(P_1) + c(P_2) = c^y(P_1) + c^y(P_2) + 2 \cdot (z_t - z_s) \leq 2 \cdot \sum_{U \in \supp(y)} y_U + 2 \cdot (z_t - z_s).$
Multiplying the terms in this bound by $\rho$ and then subtracting $(2\rho-1) \cdot (z_t - z_s)$ from both sides, we see
$(2\rho-1) \cdot (z_s - z_t) \leq \sum_{U \in \supp(y)} 2\rho \cdot y_U + z_t - z_s = OPT_{LP}.$
\end{proof}


\section{An Improved Integrality Gap Bound in Regret Metrics}\label{sec:regret}

Let $V$ be nodes and $s,t \in V$ be the start and end points. Let $c$ be  {\em symmetric} metric distances $c_{u,v} \geq 0$. For each $u,v \in V$, let $c^{\text{reg}}_{u,v} = c_{r,u} + c_{u,v} - c_{r,v}$
be the {\em regret} metric induced by $c$. It is convenient to consider a complete directed graph over $V$ where for distinct $u,v \in V$
we have $c_{u,v} = c_{v,u}$ yet $uv$ and $vu$ are themselves distinct edges: the {\em bidirected} variant of the natural undirected graph associated with $(V,c)$.
The following observations about regret metrics can be found in \cite{FriggstadS14}.
\begin{observation}
If $c$ is a metric (asymmetric or symmetric) then $c^{\text{reg}}$ is an asymmetric metric. For any $u,v \in V$ and any $u-v$ path $P$, $c(P) = c^{\text{reg}}(P) + c_{u,v}$.
For any cycle $C$, $c(C) = c^{\text{reg}}(C)$.
\end{observation}

We consider integrality gap bounds for \eqref{lp:atspp} when the metric is a regret metric. In \cite{FriggstadS17}, it was shown the integrality gap bound is 2 in the standard case $\rho = 1$ and that this is tight.
For the purpose of getting better approximations for \dl in regret metrics (i.e. the problem of minimizing the average time a node $v$ waits {\em in excess} of their shortest path distance $c_{r,v}$ from the depot), we give explicit integrality gap bounds for the more general case $1/2 < \rho \leq 1$.

Note, in the case $\rho = 1$ that the analysis from \cite{FriggstadS17} produces a stronger result. But the analysis does not extend in any clear way to the case $\rho < 1$.
We begin by recalling the following structural result by Bang-Jensen et al about decomposing preflows into branchings \cite{BangjensenFJ95}, which was made efficient
by Post and Swamy \cite{PostS15}.
\begin{theorem}[Bang Jensen et al. \cite{BangjensenFJ95}, Post and Swamy \cite{PostS15}]\label{thm:bang}
Let $D = (\{r\} \cup V, A)$ be a directed graph and $x \in \mathbb Q^A_{\geq 0}$ be a preflow. Let $\lambda_v := \min_{\{v\} \subseteq S \subseteq V} x(\delta^{in}(S))$ be the $r-v$ connectivity in $D$ under capacities $\{x_a\}_{a \in A}$. Let $K > 0$ be rational. We can obtain out-branchings $B_1, \ldots, B_q$ rooted at $r$, and rational weights $\gamma_1, \ldots, \gamma_q \geq 0$ such that $\sum_{i=1}^q \gamma_i = K, \sum_{i: q \in B_i} \gamma_i \leq x_a$ for all $a \in A$, and $\sum_{i : v \in B_i} \geq \min\{K, \lambda_v\}$ for all $v \in V$. Moreovers, such a decomposition can be computed in time that is polynomial in $|V|$ and the bit complexity of $K$ and $x$.
\end{theorem}

We require a definition and results from \cite{FriggstadS14}, some of which are adaptations from concepts in \cite{Blum07}.
\begin{definition}
Let $P$ be a path starting at $s$. For each $uv \in P$, say $uv$ is {\bf red} on $P$ if there are nodes $x, y$ on the $s-u$ portion of $P_i$ and $v-t$ portion of $i$, respectively, such that $c_{r,x} \geq c_{r,y}$. For each $v \in P$, let $\text{red}(v,P)$ be the maximal subset of red edges of the subpath of $P$ containing $v$. Note, $\text{red}(v,P)$ could be empty if $v$ is not incident to a red edge.
The {\bf red intervals} of $P$ are the maximal subpaths of its red edges.
\end{definition}
Intuitively, the red edges are part of intervals of $P$ that do not make progress toward reaching $t$. Their total $c^{\reg}$-costs can be shown to be comparable to their total $c$-costs, which is formalized as follows.
\begin{lemma}[Blum et al \cite{Blum07}]\label{lem:regbound}
For any $s-t$ path $P$, $\sum_{uv \text{ red on } P} c_{u,v} \leq \frac{3}{2} c^{\text{reg}}(P)$.
\end{lemma}

Further, if we were to keep at most one node from each maximal red interval of edges and shortcut past the other nodes, the resulting path $s = v_0, v_1, \ldots, v_k = t$ has $c_{r,v_i} < c_{r,v_{i+1}}$.
So the union of any collection of paths that are shortcut in such a way forms an acyclic graph.

Now, a solution to \eqref{lp:atspp} can be viewed as a preflow of value 1 rooted at $s$ with $\lambda_v \geq \rho$ for each $v \in V-t$ and $\lambda_t = 1$. From this observation, we round a solution using techniques from \cite{FriggstadS14}. The full description is in Algorithm \ref{alg:lat}. Here, $1/2 < \delta < \rho$ is some parameter we set later to optimize the performance of the algorithm.

{\small 
\begin{algorithm}[{\bf Rounding \eqref{lp:atspp} in regret metrics}]  \label{alg:lat}
\hrule \vspace{5pt} \

\noindent
Input: asymmetric metric $(V\cup\{r\},c^{\text{reg}})$ obtained from symmetric distances $c$.

\noindent
Output: an Hamiltonian $s-t$-rooted path $P$.
\end{algorithm}
\smallskip
\begin{enumerate}[label=D\arabic*., ref=D\arabic*, topsep=0ex, itemsep=0.5ex,
    leftmargin=*]
\item Solve \eqref{lp:atspp} to get an optimal extreme point solution $x$ with value $OPT_{LP}$.
\item Use Theorem \ref{thm:bang} to find a convex combination of out-branchings $B_1, \ldots, B_q$ rooted at $s$ and weights $\gamma_1, \ldots, \gamma_q \geq 0$ summing to 1
such that $t$ lies on each $B_i$ and each $v \in V-\{s,t\}$ lies on at least a $\rho$-fraction of these branchings. Turn each $B_i$ into a $s-t$ path $P_i$ by adding the reverse $(v,u)$ of each arc $(u,v) \in B_i$ that does not appear on the unique $s-t$ path in $B_i$ and shortcutting the resulting Eulerian $s-t$ walk
past repeated nodes. \label{step:paths}
\item Define a cut requirement function $f : 2^{V} \rightarrow \{0,1\}$ where $f(S) = 1$ if $\sum_{i : \text{red}(v, P_i) \subseteq S} \gamma_i < \delta$ for all $v \in S$.
Observe $f$ is downward-monotone: $f(S) \geq f(T)$ for sets $\emptyset \subsetneq S \subseteq T$.
Use the LP-based 2-approximation in \cite{GoemansW94} to find a forest of undirected edges $F$ such that $|\delta(S) \cap F| \geq f(S)$. Let $\mathcal C$ be the components of $F$ and let
$C_1, \ldots, C_{|C|}$ be cycles on each component of $F$ obtained by doubling and shortcutting each tree in $F$. For each cycle $C_j$ of $\mathcal C$, let $w \in C_i$ be some {\em witness node}
such that $\sum_{i : \text{red}(w, P_i) \subseteq V} \gamma_i \geq \delta$.
Let $W$ be the set of all witness over all $C_j$ (note, it could be $W \cap \{s,t\} \neq \emptyset$). View each $C_j$ as being traversed in some arbitrary direction.
\item For each $P_i$, let $P^W_i$ be the set of all nodes in $W \cap P_i$ such that all nodes of $\text{red}(w, P_i)$ are contained in the nodes of a single cycle $C_j$. Shortcut $P_i$ past nodes not in $P^W_i \cup \{s,t\}$
and call this path $P'_i$. Note the nodes of $P'_i$ lie in $W \cup \{s,t\}$.
\item View $P'_i$ with associated weights $\gamma_i/\delta$ as the path decomposition of an acyclic $s-t$ flow $z$ with value $1/\delta$ with $z(\delta(w)) \geq 1$ for each $w \in W$. Further, $z(\delta^{out}(s)) = 1/\delta < 2$. By integrality of flows with upper- and lower-bounds on each node, we may decompose $z$ as a convex combination of integral flows satisfying these bounds such that each flow supported consists of either 1 or 2 paths. Let $P$ be the cheapest
path among the flows with only one path in this decomposition. Note that $P$ is an $s-t$ path spanning all of $W$. \label{step:flow_get}
\item Complete $P$ into a Hamiltonian $s-t$ path by adding all edges of the cycles $C_i$ and shortcutting the resulting Eulerian walk. \label{step:regret_final}
\end{enumerate}
\hrule}

\begin{lemma}\label{lem:picost}
The paths $P_i$ from Step \ref{step:paths} satisfy $\sum_i \gamma_i \cdot c^{\reg}(P_i) \leq 2 \cdot OPT_{LP}$.
\end{lemma}
\begin{proof}
In \cite{FriggstadS14}, it is observed for any $s-t$ path $P$ that $c^{\text{reg}}(P) = c(P) - c_{s,t}$ and that $c(C) = c^{\text{reg}}(C)$ for any cycle $C$.
Thus, as $x$ is an $s-t$ flow with value 1 we have $OPT_{LP} = \sum_{uv} c^{\text{reg}}_{u,v} x_{u,v} = \left(\sum_{uv} c_{u,v} x_{u,v}\right) - c_{s,t}$. This can be seen by, say,
comparing the $c^{\text{reg}}$-cost with the $c$-cost of paths and cycles in a path/cycle decomposition of $x$.

Each $P_i$ is obtained by adding the reverse of each edge $uv$ of $B_i$ not on the $s-t$ path in $B_i$ (and then shortcutting the resulting Eulerian walk). Thus, $c(P_i) \leq 2 \cdot c(B_i) - c_{s,t}$
so $c^{\text{reg}}(P_i) \leq 2 \cdot (c(B_i) - c_{s,t})$. Thus,  $\sum_i \gamma_i \cdot c^{\reg}(P_i) \leq 2 \cdot \sum_i \gamma_i \cdot (c(B_i) - c_{s,t}) = \left(2 \cdot \sum_i \gamma_i \cdot c(B_i)\right) - 2 \cdot c_{s,t}$.
Now, the convex combination of the $B_i$ is dominated by $x$, so $\sum_i \gamma_i \cdot c(B_i) \leq \sum_e x_e \cdot c_e$. Finally, as $x$ constitutes one unit of $s-t$ flow, the $c$-cost of $x$ differs from the $c^{\reg}$-cost of $x$ exactly by $c_{s,t}$, so we finally see $\sum_i \gamma_i \cdot c^{\reg}(P_i) \leq 2 \cdot OPT_{LP}$.
\end{proof}

The proofs of the following two lemmas proceed in a way that is very similar to related results \cite{FriggstadS14} (though, their end goal was quite different). 
\begin{lemma}\label{lem:cyclecost}
In Step \ref{step:paths}, the function $f$ is downward-monotone and $\sum_{j} c^{\text{reg}}(C_j) \leq \frac{6}{\rho-\delta} OPT_{LP}$.
\end{lemma}
\begin{proof}
That $f$ is downward monotone is direct from the definition. We construct a vector $x'$ over edges the undirected complete graph with nodes $V$ with edge costs $c$.
That is, for each undirected edge $uv$ let $x'_{uv} = \frac{1}{\rho-\delta} \sum_{\substack{i : uv \text{ or } vu \\ \text{ is red on } P_i}} \gamma_i$. We first claim $x'(\delta(S)) \geq f(S)$ for each $\emptyset \subsetneq S \subseteq V$.
That is, suppose $S$ is such that $f(S) = 1$ and let $v$ satisfy $\sum_{i : \text{red}(v, P_i) \subseteq V} \gamma_i < \delta$. Since $v$ lies on a $\rho$-fraction of paths in total,
this means a $(\rho-\delta)$-fraction of paths $P_i$ have some edge of $\text{red}(v, P_i)$ crossing $S$, as required.

From Lemma \ref{lem:regbound}, the total $c$-cost of all red edges on $P_i$ is at most $\frac{3}{2} c^{\text{reg}}(P_i)$. Thus,
$\sum_{uv} c_{u,v} x'_{uv} \leq \frac{3}{2} \frac{1}{\rho-\delta} OPT_{LP}$. From using the LP-based 2-approximation in \cite{GoemansW94}, the $c$-cost of the result forest is then at most
$\frac{3}{\rho-\delta} OPT_{LP}$. By doubling the edges to get the cycles $C_j$, $\sum_{j} c(C_j) \leq \frac{6}{\rho-\delta} OPT_{LP}$. Finally, we chose an arbitrary direction for traversing each $C_j$
but the $c^{\text{reg}}$-cost of a directed cycle is the same as its $c$-cost, so the result follows.
\end{proof}

\begin{lemma} \label{lem:acyclic_flow}
The graph over $V$ with edges $\cup_{i=1}^q P'_i$ is an acyclic graph. Further, for each $w \in W$ we have $\sum_{i : w \text{ lies on } P'_i} \gamma_i \geq \delta$.
Finally, $\sum_{i=1}^q c^{\text{reg}}(P'_i) \leq 2 \cdot OPT_{LP}$.
\end{lemma}
\begin{proof}
We claim that we do not keep two nodes from any red interval for each $P_i$ when we form $P'_i$. But this is immediate from the fact that no cycle $C_j$ contains two nodes of $W$.

By the definition of red intervals, any path $P'$ obtained from a path $P$ by shortcutting past all but one node in each red interval yields has its nodes appearing in strictly distance-increasing order.
So, the $P'_i$ paths all start at the same location, all end at the same location, and their internal nodes strictly increase in distance from $s$.
So the union of all $P'_i$ is an acyclic graph.

Now, consider some $w \in W$ and say it lies on cycle $C_j$. At least a $\delta$-fraction of paths $P_i$ spanning $w$ satisfy $\text{red}(w, P_i) \subseteq C_j$ because $f(V(C_j)) = 0$, so each $w \in W$
lies on at least a $\delta$-fraction of paths $P'_i$.

Since $P'_i$ are obtained by shortcutting nodes from $P_i$, $\sum_{i=1}^q c^{\text{reg}}(P'_i) \leq \sum_{i=1}^q c^{\text{reg}}(P_i) \leq 2 \cdot OPT_{LP}$ by Lemma \ref{lem:picost}.
\end{proof}

We now describe how to complete the analysis.
\begin{lemma}\label{lem:penultimate}
In Step \ref{step:flow_get}, the flow $z$ has acyclic support, sends $1/\delta$ units of flow from $s$ to $t$, and has $z(\delta^{in}(w)) \geq 1$ for each $w \in W$.
The resulting path $P$ has cost $\frac{2}{2\delta-1} \cdot OPT_{LP}$.
\end{lemma}
\begin{proof}
We have $\sum_i \gamma_i/\delta = 1/\delta$. As each $P'_i$ is an $s-t$ flow, we have $z$ given by $z_{uv} = \sum_{i : uv \in P_i} \gamma_i/\delta$ is an $s-t$ flow of value $1/\delta$.
Then by Lemma \ref{lem:acyclic_flow}, the support of $z$ is acyclic, $z(\delta^{in}(w)) \geq 1$ for each $w \in W$, and $\sum_{uv} c^{\text{reg}}_{u,v} z_{uv} \leq \frac{2}{\delta} \cdot OPT_{LP}$.

By integrality of flows with integral lower- and upper-bounds on the flow through each vertex, $z$ may be decomposed into a convex-combination of integral flows $f$ satisfying the lower-bound $f(\delta^{in}(w)) \geq 1$ for each $w \in W$ and $1 \leq f(\delta^{out}(s)) \leq 2$.
Furthermore, the fraction of these flows $f$ with $f(\delta^{out}(s)) = 1$ is exactly $2-1/\delta$, so the $c^{\reg}$-cost of one such flow is at most $\frac{1}{2-1/\delta} \frac{2}{\delta} \cdot OPT_{LP} = \frac{2}{2\delta-1} \cdot OPT_{LP}$. Such a flow $f$ has no cycles because the support of $z$ is acyclic, so the edges supported by $f$ form an $s-t$ path spanning all $w \in W$.
\end{proof}

The final path is formed from grafting the cycles $C_1, \ldots, C_{|\mathcal C|}$ into $P$, so the above results yield the following.
\begin{theorem}
The final path computed in Step \ref{step:regret_final} is a Hamiltonian $s-t$ path with $c^{\text{reg}}$-cost at most $\left(\frac{6}{\rho-\delta} + \frac{2}{2\delta-1}\right) \cdot OPT_{LP}$.
\end{theorem}
\begin{proof}
By Lemma \ref{lem:penultimate}, the path $P$ is an $s-t$ path spanning $W$ with $c^{\reg}$-cost at most $\frac{2}{2\delta-1} \cdot OPT_{LP}$. Each cycle $C_j$ over a component in $\mathcal C$ contains
precisely one node in $W$, so the graph $P \cup_{j=1}^{|\mathcal C|} C_j$ has an Eulerian $s-t$ walk that visits all nodes.
By Lemma \eqref{lem:cyclecost}, the total $c^{\reg}$-cost of all cycles is at most $\frac{6}{\rho-\delta} \cdot OPT_{LP}$. The result follows because shortcutting this Eulerian walk to get a Hamiltonian path does not increase the cost of the walk, by the triangle inequality.
\end{proof}

By setting $\delta = \frac{(2\sqrt 6 - 1) \cdot \rho + 6 - \sqrt 6}{10}$ (which optimizes the parameter), we get our main result showing the integrality gap is at most $\frac{300}{42-12\sqrt6} \cdot \frac{1}{2\rho-1} \approx \frac{23.8}{2\rho-1}$.

\section{Conclusion}
We have presented the first constant-factor approximation for \dl in quasi-polynomial time by making two key contributions over the previous work on \dl in \cite{FriggstadSS13}. First, we showed the integrality
gap of \eqref{lp:atspp} is bounded by a constant for any $1/2 < \rho < 1$. Second, we showed how a combination of guesswork along with using a time-indexed relaxation can be used to overhead in stitching
from $O(\log n)$ to $O(1)$.
Naturally, the main problem is to get a constant-factor approximation in polynomial time. Perhaps a stronger set of constraints
could be added to some LP-relaxation that would help bound the cost of stitching paths together. Our approach did this by relying on our guesswork.

Additionally, even in an ideal setting where $\alpha_\rho$ is small, say $\frac{2}{2\rho-1}$, our approach would yield a 79.2-approximation (using optimal parameter $\rho = 0.725$) for the underlying \dl instance. So it would also be interesting to improve the dependence on $\alpha_\rho$ in a \dl approximation.



\appendix

\section{Reduction to Instances with Polynomially-Bounded Integer Distances}\label{app:scaling}

\begin{proof}[Proof of Theorem \ref{thm:scaling}]
Compute a value $\nu$ such that $OPT \leq \nu \leq n^2 \cdot OPT$ where $OPT$ is the optimum solution to the given \dl instance. For example, $\nu$ could be the smallest value 
such that all nodes can be covered by a single walk in the graph $G_\nu = (V+r, E_\nu)$ consisting of directed edges $E_\nu = \{uv : c_{u,v} \leq \nu\}$.
This can be checked, for example, by contracting the strongly-connected components of $G_\nu$ and checking if topologically sorting the resulting directed, acyclic graph
results in a single chain of components.

Now, the case $OPT = 0$ can detected in polynomial time as this is equivalent to checking if the strongly-connected components of the graph using only distance-0 edges forms a chain.
So we assume $OPT > 0$, thus $\nu > 0$. We then assume $c_{u,v} \geq \epsilon \cdot \nu/{n^3}$ by increasing any distance that is smaller to this amount: the distances remain metric
and the latency of any node on the optimum solution increases by at most $n \cdot \nu \leq \epsilon \cdot OPT/n$, so the total latency increases by at most $\epsilon \cdot OPT$.

Next, we may assume all distances satisfy $c_{u,v} \leq (\alpha(n) + 2\epsilon) \cdot \nu$ for the following reason. Suppose we update each distance $c_{u,v} > (\alpha(n) + 2\epsilon) \cdot \nu$ with $c_{u,v} = (\alpha(n) + 2\epsilon) \cdot \nu$.
It is easy to check these updated distances also form a metric. 
The optimum solution cost is still $OPT$ because no edge used by the optimum solution has its length shortened (as $\nu \geq OPT$).
Also, note a solution $P$ with $c(P) \leq (\alpha(n) + \epsilon) \cdot OPT$ will only use edges $uv$ where $c_{u,v} < (\alpha(n) + 2\epsilon) \cdot \nu$. So an $(\alpha+\epsilon)$-approximation
in the metric with these truncated distances yields an $(\alpha+\epsilon)$-approximation for the original distances.

Next, for all $u,v \in V+r$ let $d''(u,v) = \left\lfloor c_{u,v} \cdot \frac{n^4}{\nu \cdot \epsilon} \right\rfloor$. Let $d'$ be the shortest path metric using edge distances given by $d''$.
Let $OPT'$ denote the optimum solution to \dl instance with distances $d'$. Observe
\[ d'(u,v) \leq d''(u,v) \leq \frac{n^4}{\nu \cdot \epsilon} c_{u,v}. \]
Furthermore, $c_{u,v} \leq (\alpha(n) + 2\epsilon) \cdot \nu$ for each edge $uv$ means $d'(u,v) \leq \frac{n^4}{\epsilon} \cdot (\alpha(n) + \epsilon)$. So all distances under $d'$ are polynomially-bounded integers.
We also have $OPT' \leq \frac{n^4}{\nu \cdot \epsilon} \cdot OPT$ simply by consider an optimum solution to the original instance, but under the new distances $d'$.

Now consider a solution $P$ with $d'(P) \leq \alpha(n) \cdot OPT'$. As $d'$ is a metric, we may assume $P$ is a Hamiltonian path so $P$ traverses $n$ edges. By replacing each edge in $P$
with its shortest path using distances $d''$, we obtain a walk $W$ with $d''(W) = d'(P) \leq \alpha(n) \cdot OPT'$.

For each edge $uv$, we have $d''(u,v)+ 1 \geq c_{u,v} \cdot \frac{n^4}{\nu \cdot \epsilon}$. So the cost of $W$ under $d$ can be bounded as follows where sums over edges in $W$
include as many terms of $uv$ as its multiplicity in $W$.
\begin{eqnarray*}
c(W) & \leq & \frac{\epsilon \cdot \nu}{n^4} \cdot \sum_{uv \in W} (d''(u,v) + 1) \\
& = & \frac{\epsilon \cdot \nu}{n^4} \cdot  \left(d''(W) + |W|\right) \\
& \leq & \frac{\epsilon \cdot \nu}{n^4} \cdot (\alpha(n) \cdot OPT' + |W|) \\
& \leq & \alpha(n) \cdot OPT + \frac{\epsilon \cdot \nu}{n^2} \\
& \leq & (\alpha(n) + \epsilon) \cdot OPT.
\end{eqnarray*}
The last two bounds use $|W| \leq n \cdot |P| \leq n^2$ and $\nu \leq n^2 \cdot OPT$.
\end{proof}

\section{A Bad Example for \ref{lp:atspp}}\label{sec:lower}

We show that the dependence on the factor $\frac{1}{2\rho-1}$ in our analysis of the integrality gap of \eqref{lp:atspp} is asymptotically tight.
\begin{proof}[Proof of Theorem \ref{thm:lower}]
Consider the following metric depicted in Figure \eqref{fig:example}, which is essentially the same example showing the integrality gap is unbounded if $\rho = 1/2$ from \cite{FriggstadSS13}. The solid edges have cost 0 and the dashed edges have cost 1. The cost of all other edges not depicted is the shortest path distance in this graph (using a cost of 1 if there is no path in this graph). The number beside each edge $uv$ indicates the value of $x_{u,v}$. It can be easily check that
this is a feasible solution for \eqref{lp:atspp} even if we added the constraints $x(\delta^{in}(v)) = 1$ for each $v \in V-\{s,t\}$.
An optimal integral solution must use an edge with cost 1, yet this LP solution only has cost $2\rho-1$ so the integrality gap of \eqref{lp:atspp} is at least $\frac{1}{2\rho-1}$.
\end{proof}

\begin{figure}[h]
\begin{center}
\includegraphics[width=0.5\textwidth]{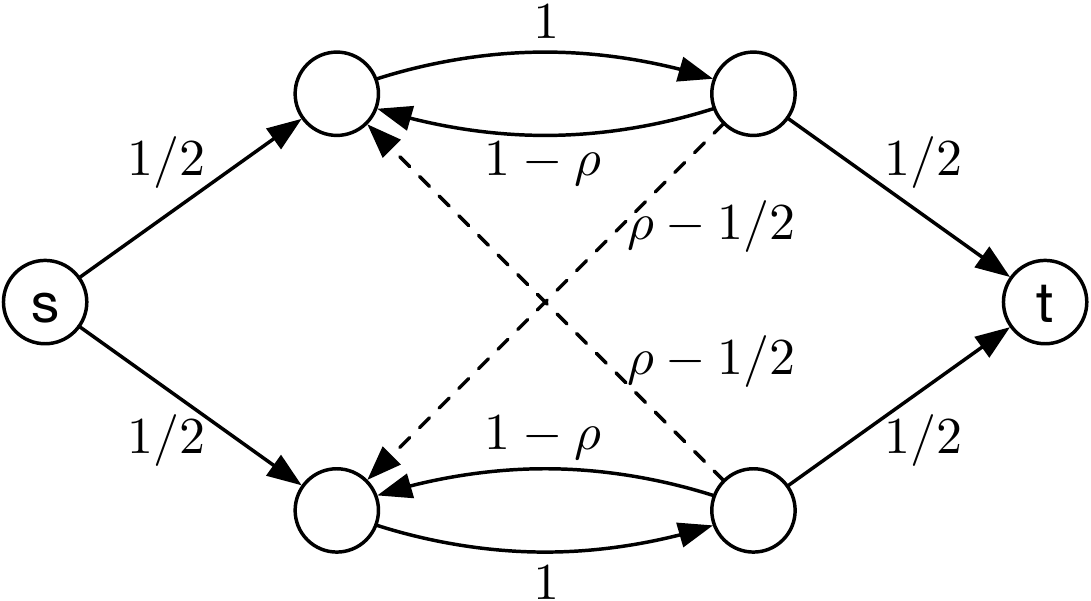}
\caption{The bad integrality gap example for \ref{lp:atspp}.}
\label{fig:example}
\end{center}
\end{figure}

\end{document}